\documentstyle[amsfonts,12pt]{article}
\title{}\date{}

\title{Exact Kink Solitons in a Monopole Confinement Problem}

\author{Shouxin Chen\\Institute of Contemporary Mathematics\\School of Mathematics\\Henan University\\
Kaifeng, Henan 475004, PR China\\ \\Yijun Li\\School of Mathematics\\Henan University\\
Kaifeng, Henan 475004, PR China\\ \\Yisong Yang \\Department of Mathematics\\Polytechnic Institute of New York University\\Brooklyn, New York 11201, USA}

\def\XXint#1#2#3{{\setbox0=\hbox{$#1{#2#3}{\int}$}
 \vcenter{\hbox{$#2#3$}}\kern-.5\wd0}}

\newtheorem{oldtheorem}{Theorem}
\newtheorem{oldassertion}[oldtheorem]{Assertion}
\newtheorem{oldproposition}[oldtheorem]{Proposition}
\newtheorem{oldremark}[oldtheorem]{Remark}
\newtheorem{oldlemma}[oldtheorem]{Lemma}
\newtheorem{olddefinition}[oldtheorem]{Definition}
\newtheorem{oldclaim}[oldtheorem]{Claim}
\newtheorem{oldcorollary}[oldtheorem]{Corollary}

\newenvironment{theorem}{\begin{oldtheorem}$\!\!\!${\bf.}}{\end{oldtheorem}}

\newenvironment{lemma}{\begin{oldlemma}$\!\!\!${\bf.}}{\end{oldlemma}}

\newbox\qedbox
\newenvironment{proof}{\smallskip\noindent{\bf Proof.}\hskip \labelsep}%
                        {\hfill\penalty10000\copy\qedbox\par\medskip}

\newcommand{\dd}{\mbox{d}}
\newcommand{\ee}{\end{equation}}
\newcommand{\be}{\begin{equation}}\newcommand{\bea}{\begin{eqnarray}}
\newcommand{\eea}{\end{eqnarray}}
\newcommand{\e}{\mbox{e}}
\newcommand{\Om}{\Omega}

\newcommand{\nn}{\nonumber}

\newcommand{\lm}{\lambda}

\begin{document}
\maketitle
\begin{abstract}
We explicitly construct all kink solitons arising in the recent study of Auzzi, Bolognesi, and Shifman
of a monopole confinement problem in ${\cal N}=2$ supersymmetric QCD. In particular, we show that all finite-energy
kink solitons must be BPS.
\end{abstract}

\maketitle

Monopole confinement in the context of supersymmetric gauge field theories 
\cite{SW,MY,Auzzi,HT,ShY,Tong,Gr} is an actively pursued subject, exploring the initial proposal by 
Mandelstam \cite{Man1,Man2}, Nambu \cite{Nambu}, and 't Hooft \cite{tH1,tH2}, who attempted to  
gain some conceptual understanding of the quark confinement problem in QCD, known to be an outstanding puzzle in theoretical physics, through a vortexline or string interaction mechanism. In the recent interesting study of Auzzi, Bolognesi, and Shifman \cite{ABS},
kink solitons 
arising in ${\cal N}=2$ supersymmetric theory with the gauge group $U(2)$ and two flavors of quarks
are formulated and described numerically which interpolate several pairs of confined monopole vacua 
through 2-strings
and
are expressed in terms of two monopole moduli space coordinates 
called profile functions
(there are four coaxial 2-string moduli space coordinates but two are irrelevant for the monopole problem).
 The purpose of this note is to obtain all these finite-energy kinks explicitly. First, we prove that
any finite-energy solution of the Euler--Lagrange equations of the kink energy of Auzzi, Bolognesi, and Shifman
\cite{ABS} must be BPS \cite{Bo,PS}. Then we 
present all the BPS solutions explicitly. These exact solutions are shown to depend precisely on two free parameters.

Following \cite{ABS}, we use $\kappa$ and $\alpha$  to denote the two monopole moduli space coordinates which are functions
of a single variable $x\in(-\infty,\infty)$.
Then the kink energy is given by the functional
\be\label{1}
E(\kappa,\alpha)=\int_{-\infty}^\infty  r\left\{4A(\kappa')^2+2(1-\kappa^{2})^{2}(\alpha')^{2} + V(\kappa,\alpha)\right\}\,\dd x,
\ee
where
\be\label{2}
V(\kappa,\alpha)=2m^{2}(1-\kappa^{2})^{2}\sin^{2}2\alpha +\frac{4m^{2}\kappa^{2}(1-\kappa^{2})^{2}\cos^{2}2\alpha}{A}
\ee
is the potential density, $'$ denotes derivative with respect to $x$, and $r,A,m$ are positive parameters.
The associated Euler--Lagrange equations  are seen to be
\bea 
\kappa''&=&-\frac{\kappa(1-\kappa^2)}A\left((\alpha')^2+m^2\sin^22\alpha-\frac{m^2(1-3\kappa^2)\cos^2 2\alpha}A\right),\label{EL1}\\
((1-\kappa^2)^2\alpha')'&=&m^2(1-\kappa^2)^2\left(1-\frac{2\kappa^2}A\right)\sin4\alpha.\label{EL2}
\eea

Kinks are finite-energy solutions of these equations satisfying the boundary condition \cite{ABS}
\be \label{bc}
\kappa(-\infty)=\kappa(\infty)=0;\quad\alpha(-\infty)=0,\quad\alpha(\infty)=\frac\pi2,
\ee
which are difficult to obtain directly. Fortunately, Auzzi, Bolognesi, and Shifman \cite{ABS} find that one may follow the
BPS trick \cite{Bo,PS} to rewrite the energy functional (\ref{1}) into the form
\bea 
E(\kappa,\alpha)&=&\int_{-\infty}^\infty r\left\{4A\left(\kappa'-\frac{m}A\kappa(1-\kappa^2)\cos2\alpha\right)^2
+2(1-\kappa^2)^2(\alpha'-m\sin2\alpha)^2\right\}\,\dd x\nn\\
&&-2mr \int_{-\infty}^\infty((1-\kappa^2)^2\cos2\alpha)'\,\dd x,
\eea
which in view of the boundary condition (\ref{bc}) gives rise to the energy lower bound
\be 
E(\kappa,\alpha)\geq 4mr.
\ee
Such a lower bound is attained when $(\kappa,\alpha)$ satisfies the BPS equations 
\bea 
\kappa'&=&\frac{m}A\kappa(1-\kappa^2)\cos2\alpha,\label{BPS1}\\
\alpha'&=& m\sin2\alpha.\label{BPS2}
\eea

It is straightforward to check that (\ref{BPS1})--(\ref{BPS2}) imply (\ref{EL1})--(\ref{EL2}). It will take some effort,
however, to show that the converse is also true. In other words, we shall prove that any solution of (\ref{EL1})--(\ref{EL2})
subject to the boundary condition (\ref{bc}) and finite-energy condition
\be \label{Ef}
E(\kappa,\alpha)<\infty,
\ee
will also satisfy the BPS equations (\ref{BPS1})--(\ref{BPS2}). Hence all kinks are necessarily BPS.

We split our proof into several steps in the form of lemmas. In doing so, we assume $(\kappa,\alpha)$ is a finite-energy solution of (\ref{EL1})--(\ref{EL2})
under the boundary condition (\ref{bc}).

\begin{lemma} \label{lemma1} For any point $a\in(-\infty,\infty)$ satisfying 
\be \label{bd}
|\alpha(x)|<\frac\pi8,\quad \kappa^2(x)<\max\left\{1,\frac A2\right\},\quad x<a,
\ee
we have
\be 
\alpha(x)
\neq 0,\quad -\infty<x<a.
\ee
\end{lemma}

\begin{proof} 
Suppose there is some $x_0<a$ such that $\alpha(x_0)=0$. Since $\alpha(-\infty)=0$,  there is a point $b<x_0$ such that $b$ is either a
local maximum or a local minimum point of $\alpha$. In either case, $\alpha'(b)=0$ but $\alpha(b)\neq0$
otherwise the uniqueness theorem for the initial value problem of ordinary differential
equations applied to (\ref{EL2}) implies $\alpha\equiv0$ which is false.
Of course, we may assume that $x_0$ is the first zero of $\alpha$ above $b$ and
$\alpha'(x_0)\neq0$.  Hence we have the alternatives
\be \label{13}
\alpha'(x_0)>0\quad\mbox{if } \alpha(b)<0;\quad \alpha'(x_0)<0\quad \mbox{if }\alpha(b)>0,
\ee
and $\alpha(x)$ does not change sign for $x\in(b,x_0)$. Integrating (\ref{EL2}) over $(b,x_0)$ and using 
(\ref{bd}), we have
\be 
(1-\kappa^2(x_0))\alpha'(x_0)=\int_{b}^{x_0}m^2(1-\kappa^2)^2\left(1-\frac{2\kappa^2}A\right)\sin4\alpha\, \dd x,
\ee
whose sign is the same as $\alpha(b)$, which contradicts (\ref{13}). Therefore we have shown that $\alpha(x)$ does not vanish for $x<a$.
\end{proof}

Since $\kappa(-\infty)=0$, the finite-energy condition already indicates 
\be\label{lm}
\liminf_{x\to-\infty}|\alpha'(x)|=0.
\ee
For our purpose, however, we need to strength this result into 

\begin{lemma}\label{lemma2} We have the full limit
\be 
\lim_{x\to-\infty} \alpha'(x)=0.
\ee
\end{lemma}
\begin{proof} 
Using Lemma \ref{lemma1} and (\ref{EL2}), we see that the quantity $((1-\kappa^2)^2\alpha')'$ does not change sign
for $x<a$. Thus $((1-\kappa^2)^2\alpha')(x)$ is monotone for $x<a$. Combining this observation with
the facts $\kappa(-\infty)=0$ and (\ref{lm}), we see that the lemma follows.
\end{proof}

\begin{lemma} For the function $\kappa$, we also have the full limit
\be 
\lim_{x\to-\infty}\kappa'(x)=0.
\ee
\end{lemma}

\begin{proof} Using Lemma \ref{lemma2}, we see that there is a point $c\in(-\infty,\infty)$ such that
\be 
\left(\frac{m^2(1-3\kappa^2)\cos^2 2\alpha}A-(\alpha')^2-m^2\sin^22\alpha\right)(x)>0,\quad 1-\kappa^2(x)>0,\quad x<c.
\ee
This result allows us to get from (\ref{EL1}) the equation
\be \label{kdd}
\kappa''=C_0(x)\kappa,\quad C_0(x)>0,\quad x<c.
\ee
Hence $\kappa(x)\neq0$ for $x<c$ or $\kappa\equiv0$ otherwise it will conflicts with the maximum principle in view of the boundary condition
$\kappa(-\infty)=0$. Assuming $\kappa\not\equiv0$ and applying $\kappa(x)\neq0$ and (\ref{kdd}), we see that $\kappa'(x)$ is monotone.
In view of this result and the finite-energy condition, we see that the proof follows.
\end{proof}

To proceed further, we consider the quantities
\bea 
P&=&(1-\kappa^2)^2\alpha'-(1-\kappa^2)^2 m\sin2\alpha,\\
Q&=&\kappa'-\frac{m}A\kappa(1-\kappa^2)\cos2\alpha,
\eea
in terms of a solution pair $(\kappa,\alpha)$ of (\ref{EL1})--(\ref{EL2}) under the boundary condition (\ref{bc}) and
the finite-energy condition (\ref{Ef}). Lemmas 2--3 and (5) imply that
\be \label{PQB}
\lim_{x\to-\infty}P(x)=0,\quad \lim_{x\to-\infty}Q(x)=0.
\ee

We now use (\ref{PQB}) to establish the following fact.

\begin{lemma}\label{lemma4} Actually we have $P\equiv0$ and $Q\equiv0$.
\end{lemma}

\begin{proof} In view of (\ref{EL1})--(\ref{EL2}), we obtain the following differential equations fulfilled by
the pair $(P,Q)$:
\bea 
P'&=&-2m \cos2\alpha P +4m\kappa(1-\kappa^2)\sin2\alpha Q,\label{P}\\
Q'&=&-\frac1{A(1-\kappa^2)^3}(\kappa P^2+m(1-\kappa^2)^3(1-3\kappa^2)\cos2\alpha Q).\label{Q}
\eea
Thus, we have
\bea \label{24}
(P^2+Q^2)'&=&-4m\cos2\alpha \,P^2+2\left(4m\kappa(1-\kappa^2)\sin2\alpha-\frac{\kappa P}{A(1-\kappa^2)^3}\right)PQ\nn\\
&&-\frac{2m}A(1-3\kappa^2)\cos2\alpha\, Q^2\nn\\
&\equiv& -\Omega(P,Q),
\eea 
where $\Omega$ may be identified with a `quadratic form' which is represented by the field-dependent matrix
\be 
M(x)=\left(\begin{array}{cc}4m\cos2\alpha&-4m\kappa(1-\kappa^2)\sin2\alpha+\frac{\kappa P}{A(1-\kappa^2)^3}\\
-4m\kappa(1-\kappa^2)\sin2\alpha+\frac{\kappa P}{A(1-\kappa^2)^3}&\frac{2m}A(1-3\kappa^2)\cos2\alpha\end{array}\right),
\ee
so that
\be 
\Omega(P,Q)=\left(\begin{array}{cc}P\\Q\end{array}\right)^\tau M(x) \left(\begin{array}{cc}P\\Q\end{array}\right).
\ee

In view of (\ref{bc}) and (\ref{PQB}), we have
\be 
\lim_{x\to-\infty}M(x)=\left(\begin{array}{cc}4m&0\\0&\frac{2m}A\end{array}\right).
\ee
Therefore, we can find some $x_0\in (-\infty,\infty)$ and constants $0<\lm_1<\lm_2<\infty$ such that
\be \label{27}
\lm_1 (P^2+Q^2)\leq \Om(P,Q)\leq\lm_2(P^2+Q^2),\quad x\leq x_0.
\ee
Inserting (\ref{27}) into (\ref{24}), we arrive at the inequality
\be \label{28}
-\lm_2(P^2+Q^2)\leq (P^2+Q^2)'\leq -\lm_1(P^2+Q^2),\quad x\leq x_0.
\ee

If $(P^2+Q^2)(x_0)>0$, we can integrate (\ref{28}) to obtain
\be
(P^2+Q^2)(x_0)\e^{\lm_1(x_0-x)}\leq (P^2+Q^2)(x)\leq (P^2+Q^2)(x_0)\e^{\lm_2(x_0-x)},\quad x<x_0.
\ee
Letting $x\to-\infty$, we have $(P^2+Q^2)(x)\to\infty$, contradicting Lemmas 2--3 and (\ref{bc}), which indicate that
$P(-\infty)=Q(-\infty)=0$.

If $P(x_0)=Q(x_0)=0$, we may use this condition in the coupled system of the first-order equations
(\ref{P}) and (\ref{Q}) and the uniqueness theorem for the initial value problems of ordinary differential equations
to infer that $P\equiv0$ and $Q\equiv0$ so that the proof of the lemma follows.
\end{proof}

We can now establish

\begin{theorem} In the context of finite-energy solutions satisfying the boundary condition
(\ref{bc}), the Euler--Lagrange equations (\ref{EL1})--(\ref{EL2}) of
the kink soliton energy (\ref{1}) and the BPS equations (\ref{BPS1})--(\ref{BPS2})
are equivalent. Thus, $(\kappa,\alpha)$ is a solution with a nontrivial $\kappa$-component if and only if $\kappa(x)\neq0$,
that is either $\kappa(x)>0$ or $\kappa(x)<0$, for all $x\in(-\infty,\infty)$.
\end{theorem}

\begin{proof} Let $(\kappa,\alpha)$ be a solution pair. Then Lemma \ref{lemma4} gives us $Q\equiv0$. So (\ref{BPS1}) is
fulfilled. Hence there is no point $x_0$ such that $\kappa^2(x_0)=1$ otherwise the uniqueness theorem will imply that
$\kappa^2(x)=1$ for all $x$ which is inconsistent with the boundary condition $\kappa(-\infty)=0$. Thus $1-\kappa^2(x)
\neq0$ for any $x$. Inserting this result into the conclusion $P\equiv0$ arrived at in Lemma \ref{lemma4}, we see
that (\ref{BPS2}) is also fulfilled.

If $\kappa\not\equiv0$, then $\kappa(x)\neq0$ for any $x\in(-\infty,\infty)$ since by virtue of the equation (\ref{BPS1})
and the uniqueness theorem we deduce $\kappa\equiv0$ if there is a point $x_0$ such that $\kappa(x_0)=0$.
\end{proof}

The above theorem allows us to focus on the BPS equations (\ref{BPS1})--(\ref{BPS2}) which are upper triangular and can
be integrated readily.

In fact, integrating (\ref{BPS2}), we have
\be \label{30}
\alpha(x)=\arctan\left(c\e^{2mx}\right),\quad c>0.
\ee
Substituting (\ref{30}) into (\ref{BPS1}) with
\be 
\cos\alpha=\frac1{\sqrt{1+c^2 \e^{4mx}}}, 
\ee
and assuming $\kappa>0$, we obtain a separable equation which can be integrated to give us
\be 
\ln\frac{\kappa^2}{1-\kappa^2}=\frac{2m}A\int\frac{1-c^2 \e^{4mx}}{1+c^2\e^{4mx}}\,\dd x=\frac1A(2mx-\ln(1+c^2\e^{4mx}))+C,
\ee 
where $C$ is an integrating constant. 

It will be convenient to absorb the constant $c>0$ with an initial reference point, $x_0$, so that $c=\e^{-2m x_0}$. 
Thus, with
$\kappa>0$, we may summarize our solution into the formulas
\bea 
\alpha(x)&=&\arctan\left(\e^{2m(x-x_0)}\right),\\
\kappa(x)&=&\left(\frac{ q\sigma(x-x_0)}{1+ q\sigma(x-x_0)}\right)^{\frac12},\quad \sigma(x)=\frac{\e^{\frac{2m}A x}}{(1+\e^{4mx})^{\frac1A}},
\eea
where $q>0$ is another free parameter. Hence the explicit solution depends on two free parameters, $x_0$ and $q$.

From the structure of the equations (\ref{BPS1})--(\ref{BPS2}), we see that if $(\kappa,\alpha)$ is a solution, so is
$(-\kappa,\alpha)$. Thus we have obtained all possible finite-energy solutions of (\ref{EL1})--(\ref{EL2}) subject to the boundary 
condition (\ref{bc}).

\medskip 

{\small The authors would like to thank the referee for helpful suggestions.
The research of Chen was supported in part by
Henan Basic Science and Frontier Technology Program Funds under grant 112300410054.}

\small{

}
\end{document}